\documentclass{llncs}

\usepackage[utf8]{inputenc}
\usepackage{amsfonts}
\usepackage{amssymb,amsmath}

\usepackage{gastex}
\usepackage[usenames]{color}

\renewcommand{\AA}{\mathcal{A}}

\newcommand{\BB}{\mathcal{B}}
\newcommand{\CC}{\mathcal{C}}

\newcommand{\prob}[1]{\mathbb{P}_{#1}}
\newcommand{\set}[1]{\{ #1 \}}
\newcommand{\val}[1]{\text{val}(#1)}
\newcommand{\merge}{\textrm{merge}}
\newcommand{\chck}{\textrm{check}}
\newcommand{\apply}{\textrm{apply}}
\newcommand{\finish}{\textrm{finish}}
\newcommand{\wait}{\textrm{wait}}

\title{Pushing undecidability of the isolation problem\\
for probabilistic automata}
\author{Nathana\"el Fijalkow \inst{1}, Hugo Gimbert \inst{2} 
\and Youssouf Oualhadj \inst{3}}
\institute{
LIAFA, CNRS \& Universit\'e Denis Diderot - Paris 7, France \\
\email{nath@liafa.jussieu.fr} \\
\and LaBRI, CNRS, France \\
\email{hugo.gimbert@labri.fr} \\
\and LaBRI, Universit\'e Bordeaux 1, France \\
\email{youssouf.oualhadj@labri.fr}}

\begin{document}

\maketitle

\begin{abstract}
This short note aims at proving that the isolation problem 
is undecidable for probabilistic automata with only one probabilistic transition.
This problem is known to be undecidable for general probabilistic automata, 
without restriction on the number of probabilistic transitions.
In this note, we develop a simulation technique that allows to simulate 
any probabilistic automaton with one having only one probabilistic transition.
\end{abstract}

\section{Introduction}

\noindent{\bf Probabilistic automata.}
Rabin introduced probabilistic automata over finite words 
as a natural and simple computation model~\cite{rabinsem}.
A probabilistic automaton can be thought as a non-deterministic automaton,
where non-deterministic transitions are chosen 
according to a fixed probabilistic distribution.
Probabilistic automata drew attention and have been extensively studied 
(see~\cite{bukharaev} for a survey).

\noindent{\bf The isolation problem.}
However, on the algorithmic side, most of the results are undecidability results.
The isolation problem asks, given some probability $0 \leq \lambda \leq 1$, 
whether there exists words accepted with probability arbitrarily close to $\lambda$.
Bertoni showed that this problem is undecidable~\cite{bertoni1,bertoni2}.

\noindent{\bf Contribution.}
In this note, we prove that the isolation problem is undecidable, 
even for probabilistic automata having only one probabilistic transition.
To do this, we develop a simulation technique that allows 
to simulate any probabilistic automaton
with one having only one probabilistic transition.

\noindent{\bf Outline.} 
Section 2 is devoted to definitions. In section 3, we develop a simulation technique,
which allows to simulate any probabilistic automaton 
with one having only one probabilistic transition.
Using this technique we show that the isolation problem is undecidable 
for this very restricted class of automata.

\section{Definitions}

Given a finite set of states $Q$, a probability distribution (distribution for short) over $Q$ 
is a row vector $\delta$ of size $|Q|$ with rational entries in $[0,1]$ 
such that $\sum_{q \in Q} \delta(q) = 1$.
We denote by $\delta_q$ the distribution such that 
$\delta_q(q') = 1$ if $q' = q$ and $0$ otherwise.
A probabilistic transition matrix $M$ is a square matrix of size $|Q| \times |Q|$,
such that for a state $s$, $M_a(s,\_) $ is a distribution over $Q$.

\begin{definition}[Probabilistic automaton]
A probabilistic automaton is a tuple $\AA = (Q, A, (M_a)_{a \in A}, q_0, F)$, 
where $Q$ is a finite set of states, $A$ is the finite input alphabet, 
$(M_a)_{a \in A}$ are the probabilistic transition matrices,
$q_0$ is the initial state and $F$ is the set of accepting states.
\end{definition}

For each letter $a \in A$, $M_a(s,t)$ is the probability to go from state $s$ to state $t$ 
when reading letter $a$.
A probabilistic transition is a couple $(s,a)$ such that $M_a (s,t) \notin \set{0,1}$ 
for some $t$.

A probabilistic automaton is said \textit{simple} if for all $a$, for all states $s$ and $t$,
we have $M_a(s,t) \in \set{0,\frac{1}{2},1}$.

Given an initial distribution $\delta$ and an input word $w$, we define $\delta \cdot w$ 
by induction on $w$: we have $\delta \cdot \varepsilon = \delta$,
then for a letter $a$ in $A$, we have $\delta \cdot a = M_a \cdot \delta$
and if $w = v \cdot a$, then $\delta \cdot (v \cdot a) = (\delta \cdot v) \cdot a$.

We denote by $\prob{\AA}(s \xrightarrow{w} T)$ the probability to reach the set $T$ 
from state $s$ when reading the word $w$, that is $\sum_{t \in T}(\delta_s \cdot w)(t)$.

\begin{definition}[Value and acceptance probability]
The \emph{acceptance probability} of a word $w \in A^*$ by $\AA$ is
$\prob{\AA}(w) = \prob{\AA}(q_0 \xrightarrow{w} F)$.
The \emph{value} of $\AA$, denoted $\val{\AA}$,
is the supremum acceptance probability: $\val{\AA} = \sup_{w \in A^*} \prob{\AA}(w)$.
\end{definition}

\section{Simulation with one probabilistic transition}

We first show how to simulate a probabilistic automaton with one having only one probabilistic transition, \textit{up to a regular language}:
\begin{proposition}
For any simple probabilistic automata $\AA = (Q,A,(M_a)_{a \in A}, q_0, F)$, 
there exists a simple probabilistic automaton $\BB$ over a new alphabet $B$, with one probabilistic transition,
and a morphism $\widehat{\_} : A^* \mapsto B^*$ such that:
$$\forall w \in A^*, \prob{\AA}(w) = \prob{\BB}(\widehat{w}).$$
\end{proposition}

The morphism $\widehat{\_}$ will not be onto, so this simulation works 
up to the regular language $\set{\widehat{w} \mid w \in A^*}$.
We shall see that the automaton $\BB$ will not be able to check that a word read belongs 
to this language, which makes this restriction unavoidable in this construction.

We first give the intuitions behind the construction.
Intuitively, while reading the word $w$, the probabilistic automaton $\AA$ 
``throw parallel threads''.
A computation of $\AA$ over $w$ can be viewed as a tree, 
where probabilistic transitions correspond to branching nodes.
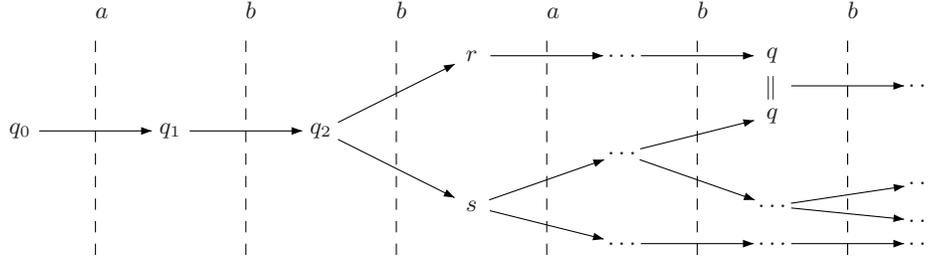
\begin{figure}
\begin{center}
\begin{picture}(120,30)(0,0)
	\gasset{Nw=5,Nh=5,Nframe=n}

	\put(10,30){$a$}
	\drawline[dash={1.5}0,AHnb=0](10,-2)(10,27)
	\put(30,30){$b$}
	\drawline[dash={1.5}0,AHnb=0](30,-2)(30,27)
	\put(50,30){$b$}
	\drawline[dash={1.5}0,AHnb=0](50,-2)(50,27)
	\put(70,30){$a$}
	\drawline[dash={1.5}0,AHnb=0](70,-2)(70,27)
	\put(90,30){$b$}
	\drawline[dash={1.5}0,AHnb=0](90,-2)(90,27)
	\put(110,30){$b$}
	\drawline[dash={1.5}0,AHnb=0](110,-2)(110,27)

  	\node(1)(0,15){$q_0$}
  	\node(2)(20,15){$q_1$}
  	\node(3)(40,15){$q_2$}
  	\node(4)(60,25){$r$}
  	\node(5)(60,5){$s$}
  	\node(6)(80,25){$\ldots$}
  	\node(7)(80,12){$\ldots$}
  	\node(8)(80,0){$\ldots$}
  	\node(9)(100,25){$q$}
  	\node(9-10)(100,21){}
  	\node(10)(100,17){$q$}
  	\node(11)(100,5){$\ldots$}
  	\node(12)(100,0){$\ldots$}
  	\node(13)(120,0){$\ldots$}
  	\node(14)(120,3){$\ldots$}
  	\node(15)(120,8){$\ldots$}
  	\node(16)(120,21){$\ldots$}

  	\drawedge(1,2){}
  	\drawedge(2,3){}
  	\drawedge(3,4){}
  	\drawedge(3,5){}
  	\drawedge(4,6){}
  	\drawedge(5,7){}
  	\drawedge(5,8){}
  	\drawedge(6,9){}
  	\drawedge(7,10){}
  	\drawedge(7,11){}
  	\drawedge(8,12){}
  	\drawedge(9-10,16){}
  	\drawedge(11,15){}
  	\drawedge(11,14){}
  	\drawedge(12,13){}

	\put(99,20){$\parallel$}
\end{picture}
\end{center}
\label{ex_computation}
\caption{An example of a computation}
\end{figure}

On the figure, reading $a$ from $q_0$ or $b$ from $q_1$ leads deterministically to the next state.
Reading $b$ from $q_2$ leads at random to $r$ or to $s$, hence the corresponding node is branching. Our interpretation is that two parallel threads are thrown.
Let us make two observations:
\begin{itemize}
	\item threads are not synchronised: reading the fourth letter (an $a$), the first thread leads deterministically to the next state, while the second thread randomizes;
	\item threads are merged so there are at most $n = |Q|$ parallel threads: whenever two threads synchronize to the same state $q$, they are merged.
	This happens in the figure after reading the fifth letter ($b$).
\end{itemize}

The automaton $\BB$ we construct will simulate the $n$ threads from the beginning, 
and take care of the merging process each step.

\begin{proof}
We denote by $q_i$ the states of $\AA$, \textit{i.e} $Q = \set{q_0, \ldots, q_{n-1}}$.
The alphabet $B$ is made of two new letters `$*$' and `$\merge$' plus, 
for each letter $a \in A$ and state $q \in Q$,
two new letters $\chck(a,q)$ and $\apply(a,q)$, so that:
$$B = \set{*, \merge} \cup \bigcup_{a \in A, q \in Q} \set{\chck(a,q), \apply(a,q)}$$

We now define the automaton $\BB$.
We duplicate each state $q \in Q$, and denote the fresh copy by $\bar{q}$.
Intuitively, $\bar{q}$ is a temporary state that will be merged at the next merging process.
States in $\BB$ are either a state from $Q$ or its copy, or one of the three fresh states $s_*$, $s_0$ and $s_1$.

The initial state remains $q_0$ as well as the set of final states remains $F$.

The transitions of $\BB$ are as follows:
\begin{itemize}
	\item for every letter $a \in A$ and state $q \in Q$, 
the new letter $\chck(a,q)$ from state $q$ leads deterministically to state $s_*$ \textit{i.e} $M_{\chck(a,q)}(q) = s_*$,
	\item the new letter $*$ from state $s_*$ leads with probability half to $s_0$ and half to $s_1$, \textit{i.e} $M_{s_*}(*) = \frac{1}{2} s_0 + \frac{1}{2} s_1$ (this is the only probabilistic transition of $\BB$);
	\item the new letter $\apply(a,q)$ from states $s_0$ and $s_1$ applies the transition function from $q$ reading $a$: 
if the transition $M_a(q)$ is deterministic, \textit{i.e} $M_a(q,r) = 1$ for some state $r$ then $M_{\apply(a,q)}(s_0) = \bar{r}$ and $M_{\apply(a,q)}(s_1) = \bar{r}$,
else the transition $M_a(q)$ is probabilistic \textit{i.e}
$M_a(q) = \frac{1}{2} r + \frac{1}{2} r'$ for some states $r,r'$, 
then $M_{\apply(a,q)}(s_0) = \bar{r}$ and $M_{\apply(a,q)}(s_1) = \bar{r'}$;
	\item the new letter $\merge$ activates the merging process: it consists in replacing $\bar{q}$ by $q$ for all $q \in Q$.
\end{itemize}
Whenever a couple (letter, state) does not fall in the previous cases, it has no effect.
The gadget simulating a transition is illustrated in the figure.

\begin{figure}
\begin{center}
\begin{picture}(60,20)(0,0)
	\gasset{Nw=5,Nh=5,Nframe=n}

  	\node(q)(0,10){$q$}
  	\node(s)(20,10){$s_*$}
  	\node(s0)(40,20){$s_0$}
  	\node(s1)(40,0){$s_1$}
  	\node(r0)(60,20){$\bar{r_0}$}
  	\node(r1)(60,0){$\bar{r_1}$}

  	\drawedge(q,s){$\chck(a,q)$}
  	\drawedge(s,s0){$*$}
  	\drawedge(s,s1){$*$}
  	\drawedge(s0,r0){$\apply(a,q)$}
  	\drawedge(s1,r1){$\apply(a,q)$}
\label{fig:gadget1}
\end{picture}
\end{center}
\end{figure}
Now we define the morphism $\widehat{\_} : A^* \mapsto B^*$ by its action on letters:
$$\widehat{a} =
\chck(a,q_0) \cdot * \cdot \apply(a,q_0) \ldots \chck(a,q_{n-1}) \cdot * \cdot \apply(a,q_{n-1}) \cdot \merge.$$

The computation of $\AA$ while reading $w$ in $A^*$
is simulated by $\BB$ on $\widehat{w}$, \textit{i.e} we have:
$$\prob{\AA}(w) = \prob{\BB}(\widehat{w})$$

This completes the proof.\hfill\qed
\end{proof}

Let us remark that $\BB$ is indeed unable to check that a letter $\chck(a,q)$ 
is actually followed by the corresponding $\apply(a,q)$: 
inbetween, it will go through $s_*$ and ``forget'' the state it was in.

%

We now improve the above construction: we get rid of the regular external condition.
To this end, we will use probabilistic automata whose transitions have probabilities $0$, $\frac{1}{3}$, $\frac{2}{3}$ or $1$.
This is no restriction, as stated in the following lemma:

\begin{lemma}
For any simple probabilistic automata $\AA = (Q,A,(M_a)_{a \in A}, q_0, F)$, 
there exists a probabilistic automaton $\BB$ whose transitions have probabilities $0$, $\frac{1}{3}$, $\frac{2}{3}$ or $1$, such that for all $w$ in $A^*$, we have:
$$\val{\AA} = \val{\BB}.$$
\end{lemma}

\begin{proof}
We provide a construction to pick with probability half, using transitions with probability
$0$, $\frac{1}{3}$, $\frac{2}{3}$ and $1$.
The construction is illustrated in the figure.
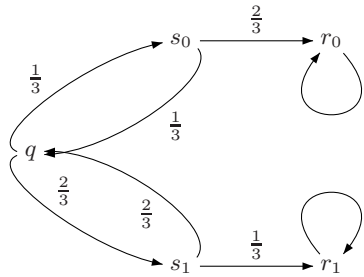
\begin{figure}
\begin{center}
\begin{picture}(60,33)(0,0)
	\gasset{Nadjust=wh,Nframe=n}

  	\node(q)(0,15){$q$}
  	\node(s0)(20,30){$s_0$}
  	\node(s1)(20,0){$s_1$}
  	\node(r0)(40,30){$r_0$}
  	\node(r1)(40,0){$r_1$}

   	\drawbpedge(q,0,-10,s0,180,10){$\frac{1}{3}$}
   	\drawbpedge(q,0,-10,s1,180,10){$\frac{2}{3}$}
  	\drawbpedge(s0,-10,10,q,-10,10){$\frac{1}{3}$}
  	\drawbpedge(s1,15,10,q,15,10){$\frac{2}{3}$}
  	\drawedge(s0,r0){$\frac{2}{3}$}
  	\drawedge(s1,r1){$\frac{1}{3}$}
  	\drawloop[loopangle=-90](r0){}
  	\drawloop(r1){}
\end{picture}
\end{center}
\end{figure}

In this gadget, the only letter read is a fresh new letter $\sharp$.
The idea is the following: to pick with probability half $r_0$ or $r_1$, 
we sequentially pick with probability a third or two thirds.
Whenever the two picks are different, if the first was a third, 
then choose $r_0$, else choose $r_1$. 
This happens with probability half each.
We easily see that $\prob{\AA}(a_0 \cdot a_1 \cdot \ldots a_{k-1}) 
= \sup_p \prob{\BB}(a_0 \cdot \sharp^p \cdot a_1 \cdot \sharp^p \ldots a_{k-1} \cdot 
\sharp^p)$.
\hfill\qed
\end{proof}

\begin{proposition}
For any simple probabilistic automata $\AA = (Q,A,(M_a)_{a \in A}, q_0, F)$, 
there exists a simple probabilistic automaton $\BB$ over a new alphabet $B$, 
with one probabilistic transition, such that:
$$\val{\AA} \geq \lambda \Leftrightarrow \val{\BB} \geq \lambda.$$
\end{proposition}

Thanks to the lemma, we assume that in $\AA$, transitions have probabilities $0$, $\frac{1}{3}$, $\frac{2}{3}$ or $1$.

We first deal with the case where $\lambda = 1$.
The new gadget used to simulate a transition is illustrated in the figure.

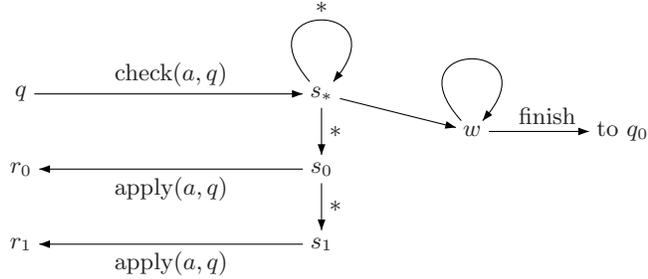
\begin{figure}
\begin{center}
\begin{picture}(60,28)(0,0)
	\gasset{Nadjust=wh,Nframe=n}

  	\node(q)(0,20){$q$}
  	\node(s)(40,20){$s_*$}
  	\node(s0)(40,10){$s_0$}
  	\node(s1)(40,0){$s_1$}
  	\node(r0)(0,10){$r_0$}
  	\node(r1)(0,0){$r_1$}

  	\node(w)(60,15){$w$}
  	\node(q_0)(80,15){to $q_0$}

  	\drawedge(q,s){$\chck(a,q)$}
	\drawloop(s){$*$}
  	\drawedge(s,w){}
	\drawloop(w){}
  	\drawedge(w,q_0){$\finish$}
  	\drawedge(s,s0){$*$}
  	\drawedge(s0,s1){$*$}
  	\drawedge(s0,r0){$\apply(a,q)$}
  	\drawedge(s1,r1){$\apply(a,q)$}
\label{fig:gadget2}
\end{picture}
\end{center}
\end{figure}

The automaton $\BB$ reads words of the form $u_1 \cdot \finish \cdot u_2 \cdot \finish \ldots$, where `$\finish$' is a fresh new letter.
The idea is to ``skip'', or ``delay'' part of the computation of $\AA$:
each time the automaton $\BB$ reads a word $u_i$, it will be skipped with some probability.

Simulating a transition works as follows: 
whenever in state $s_*$, reading two times the letter `$*$' leads 
with probability half to $s_1$, quarter to $s_0$ and quarter to $s$.
As before, from $s_0$ and $s_1$, we proceed with the simulation.
However, in the last case, we ``wait'' for the next letter `$\finish$' 
that will restart from $q_0$.
Thus each time a transition is simulated, the word being read is skipped 
with probability $\frac{1}{4}$.

Delaying part of the computation allows to multiply the number of threads.
We will use the accepted threads to check the extra regular condition we had before.
To this end, as soon as a simulated thread is accepted in $\BB$, 
it will go through an automaton (denoted $\CC$ in the construction) 
that checks the extra regular condition.

\begin{proof}
We keep the same notations.
The alphabet $B$ is made of three new letters: `$*$', `$\merge$' and `$\finish$' plus, 
for each letter $a \in A$ and state $q \in Q$,
two new letters $\chck(a,q)$ and $\apply(a,q)$, so that:
$$B = \set{*, \merge, \finish} \cup \bigcup_{a \in A, q \in Q} \set{\chck(a,q), \apply(a,q)}$$

We first define a syntactic automaton $\CC$.
We define a morphism $\widehat{\_} : A^* \mapsto B^*$ by its action on letters:
$$\widehat{a} =
\chck(a,q_0) \cdot * \cdot * \cdot \apply(a,q_0) \ldots \chck(a,q_{n-1}) 
\cdot * \cdot  * \cdot \apply(a,q_{n-1}) \cdot \merge.$$
Consider the regular language $L = \set{\widehat{w} \cdot \finish \mid w \in A^*}^*$, 
and $\CC = (Q_\CC,\delta_\CC,s_\CC,F_\CC)$ an automaton recognizing it.

We now define the automaton $\BB$.
We duplicate each state $q \in Q$, and denote the fresh copy by $\bar{q}$.
States in $\BB$ are either a state from $Q$ or its copy, 
a state from $Q_\CC$ or one of the four fresh states $s_*$, $s_0$, $s_1$ and $\wait$.

The initial state remains $q_0$, and the set of final states is $F_\CC$.

The transitions of $\BB$ are as follows:
\begin{itemize}
	\item for every letter $a \in A$ and state $q \in Q$, 
the new letter $\chck(a,q)$ from state $q$ leads deterministically to state $s_*$ \textit{i.e} $M_{\chck(a,q)}(q) = s_*$,
	\item the new letter $*$ from state $s_*$ leads with probability half to $s_*$ and half to $s_0$, \textit{i.e} $M_{s_*}(*) = \frac{1}{2} s_* + \frac{1}{2} s_0$ (this is the only probabilistic transition of $\BB$);
	\item any other letter from state $s_*$ leads deterministically to $w$,
\textit{i.e} $M_{s_*}(\_) = \wait$;
	\item the new letter $*$ from state $s_0$ leads deterministically to $s_1$, \textit{i.e} $M_{s_0}(*) = s_1$;
	\item the new letter $\apply(a,q)$ from states $s_0$ and $s_1$ applies 
the transition function from $q$ reading $a$: 
if the transition $M_a(q)$ is deterministic, 
\textit{i.e} $M_a(q,r) = 1$ for some state $r$ then $M_{\apply(a,q)}(s_0) = \bar{r}$ 
and $M_{\apply(a,q)}(s_1) = \bar{r}$,
else the transition $M_a(q)$ is probabilistic \textit{i.e}
$M_a(q) = \frac{1}{2} r + \frac{1}{2} r'$ for some states $r,r'$, 
then $M_{\apply(a,q)}(s_0) = \bar{r}$ and $M_{\apply(a,q)}(s_1) = \bar{r'}$;
	\item the new letter $\merge$ activates the merging process: 
it consists in replacing $\bar{q}$ by $q$ for all $q \in Q$;
	\item the new letter $\finish$ from state $\wait$ leads deterministically to $q_0$;
	\item the new letter $\finish$ from state $q$ in $F$ leads deterministically to $s_\CC$;
	\item the new letter $\finish$ from any other state is not defined (there is a deterministic transition to a bottom non-accepting state).
\end{itemize}
Transitions in $\CC$ are not modified.
Whenever a couple (letter, state) does not fall in the previous cases, it has no effect.

We now show that this construction is correct.

We first prove that for all $w \in A^*$, there exists a sequence of words $(w_p)_{p \geq 1}$ such that $\prob{\AA}(w) = \sup_p \prob{\BB}(w_p)$.

We have, for $\delta$ a distribution over $Q$:
$$\delta_\BB (\delta,\hat{a}) = \frac{3}{4} \delta_\AA (\delta,a) + \frac{1}{4} \wait.$$
It follows:
$$\delta_\BB (\delta,\hat{w}) = \left(\frac{3}{4}\right)^k \delta_\AA (\delta,w) 
+ 1 - \left(\frac{3}{4}\right)^k \wait,$$
where $k = |w|$.
Hence:
$$\delta_\BB (q_0,\hat{w} \cdot \finish) = 
\left(\frac{3}{4}\right)^k \prob{\AA}(w)
+ 1 - \left(\frac{3}{4}\right)^k q_0.$$

The computation of $\AA$ while reading $w$ is simulated by $\BB$ 
on $\widehat{w} \cdot \finish$. 
This implies that $\sup_p \prob{\BB}((\widehat{w} \cdot \finish)^p) = \prob{\AA}(w)$,
hence if $\val{\AA} = 1$, then $\val{\BB} = 1$.

Conversely, we prove that if $\val{\BB} = 1$, then $\val{\AA} = 1$.
Let $w$ a word read by $\BB$ accepted with probability close to $1$, we slice it as follows: 
$w = u_1 \cdot \finish \cdot \ldots \cdot u_k \cdot \finish$, 
such that $u_i$ does not contain the letter $\finish$.
The key observation is that if $k = 1$, 
the word $w$ is accepted with probability at most $\frac{3}{4}$. 
Hence we consider only the case $k > 1$.
We assume without loss of generality that $\prob{\AA}(u_1) > 0$ 
(otherwise we delete $u_1 \cdot \finish$ and proceed).
In this case, a thread has been thrown while reading $u_1$ that reached $s_\CC$, 
so the syntactic process started: 
it follows that $u_i$ for $i > 1$ are in the image of $\widehat{\_}$.
This implies that the simulation is sound: 
from $w$ we can recover a word in $A^*$ accepted 
with probability arbitrarily close to $1$ by $\AA$.

The case where $\lambda$ is any positive rational is handled similarly.
We only need to ensure that the previous key observation still holds:
a word of the form $u \cdot \finish$ where $u$ does not contain $\finish$ cannot be accepted with probability more than $\frac{3\lambda}{4}$.
This is made possible by slightly modifying the simulation gadget, adding new intermediate states.

This completes the proof.\hfill\qed
\end{proof}

We conclude:

\begin{theorem}
The isolation problem is undecidable for simple automata with one probabilistic transition.
\end{theorem}

\bibliographystyle{alpha}
\bibliography{bib}
\end{document}